\documentclass{amsart}
\usepackage[english]{babel}
\usepackage[latin1]{inputenc}
\usepackage{verbatim}
\usepackage{bbm}
\usepackage{amssymb, amsfonts, amsmath, amscd, mathrsfs, mathtools}
\usepackage{epsfig, graphics, subfig, url, color, graphicx}
\usepackage{enumerate}

\usepackage{tikz,pgfplots}
\usetikzlibrary{automata, positioning, decorations.pathreplacing, arrows, calc}

\def\N{\ensuremath{\mathbb{N}}}

\def\C{\ensuremath{\mathbb{C}}}

\def\P{\ensuremath{\mathbb{P}}}

\newtheorem{lemma}{Lemma}

\newtheorem{proposition}{Proposition}
\newtheorem{theorem}{Theorem}
\newtheorem{corollary}{Corollary}

\def\diff{\ensuremath{\mathop{}\mathopen{}\mathrm{d}}}

\newcommand{\ind}[1]{\ensuremath{\mathbbm{1}_{\left\{ #1 \right\}}}}
\newcommand{\abs}[1]{\ensuremath{\left| #1 \right|}}
\newcommand{\var}[1]{\ensuremath{\underset{#1}{\mathrm{var\;}}}}

\title{Analysis of a trunk reservation policy \\ in the framework of fog computing}

\author[F. Guillemin]{Fabrice Guillemin}
\address[F. Guillemin]{CNC/NCA Orange Labs2, Avenue Pierre Marzin, 22300 Lannion, France}
\email{Fabrice.Guillemin@orange.com}
\author[G. Thompson]{Guilherme Thompson}
\address[G. Thompson]{INRIA Paris, 2 rue Simone Iff, 75589 Paris Cedex 12, France}
\email{Guilherme.Thompson@inria.fr}

\begin{document}
\begin{abstract}
We analyze in this paper a system composed of two data centers with limited capacity in terms of servers. When one request for a single server is blocked at the first data center, this request is forwarded to the second one. To protect the single server requests originally assigned to the second data center, a trunk reservation policy is introduced (i.e., a redirected request is accepted only if there is a sufficient number of free servers at the second data center). After rescaling the system by assuming that there are many servers in both data centers and high request arrival rates, we are led to analyze a random walk in the quarter plane, which has the particularity of having non constant reflecting conditions on one boundary of the quarter plane. Contrary to usual reflected random walks, to compute the stationary distribution of the presented random walk, we have to determine three unknown functions, one polynomial and two infinite generating functions. We show that the coefficients of the polynomial are solutions to a linear system. After solving this linear system, we are able to compute the two other unknown functions and the blocking probabilities at both data centers. Numerical experiments are eventually performed to estimate the gain achieved by the trunk reservation policy.
\end{abstract}

\maketitle
\hrule
\vspace{-3mm}
\setcounter{tocdepth}{1}
\tableofcontents
\vspace{-10mm}
\hrule
\bigskip



\section{Introduction}

Fog computing \cite{Bonomi,Shenker,Rai,Wood} is considered by many actors of the telecommunication ecosystem as a major breakthrough in the design of networks for both network operators and content providers. For the former, deploying storage and computing DCs at the edge of the network enables them to reduce the load within the network and on critical links such as peering links. In addition, network operators can take benefit of these DCs to dynamically instantiate virtualized network functions. Content providers can take benefit of distributed storage and computing DCs to optimize service platform placement and thus to improve the quality experienced by end users.

Fog computing relies on distributed data centers (DCs), which are much smaller than big centralized DCs generally used in cloud computing. Because of potential resource limitation, user requests may be blocked if resources are exhausted at a (small) DC. This is a key difference with cloud computing where resources are often considered as infinite. In this paper, we consider the case of a computing resource service where users request servers available at a DC. If no server is available, then a user request may be blocked.

To reduce the blocking probability, it may be suitable that DCs collaborate. This is certainly a key issue in fog computing, which makes the design of networks and fog computing very different from that of cloud computing. Along this line of investigations, an offloading scheme has been investigated in \cite{FGRT}, where requests blocked at a DC are forwarded to another one with a given probability. In this paper, we investigate the case when a blocked request is systematically forwarded to another DC but to protect those requests which are originally assigned to that DC, a redirected request is accepted only if there is a sufficient large number of idle servers. In the framework of telephone networks, this policy is known as \emph{trunk reservation} \cite{Ross}.

In the following, we consider the case of two DCs, where the trunk reservation policy is applied in one server only; the analysis of the case when the policy is applied in both DCs is a straightforward extension of the case considered but involves much more computations. We further simplify the system by reasonably assuming that both DCs have a large number of servers. From a theoretical point of view, this leads us to rescale the system and to consider limiting processes. The eventual goal of the present analysis is to estimate the gain achieved by the trunk reservation policy. 

By considering the number of free servers in both DCs, we are led after rescaling to analyze a random walk in the quarter plane. This kind of process has been extensively studied in the technical literature (see for instance the book by Fayolle \emph{et al} \cite{FIM}). For the random walk appearing in this paper, even if the kernel is similar to that analyzed in \cite{FayolleIas} (and in \cite{FGRT}), the key difference is that the reflecting conditions on the boundaries of the quarter plane are not constant. More precisely, the reflecting coefficients in the negative vertical direction along the $y$-axis take three different values depending on a given threshold (namely, the trunk reservation threshold). 

This simple difference with usual random walks in the quarter plane makes the analysis much more challenging. Contrary to the usual case which consists of determining two unknown functions, we have in the present case to decompose one unknown function into two pieces (one polynomial and one infinite generating function) and thus to determine three unknown functions. We show that the coefficients of the unknown polynomial can be computed by solving a linear system. Once this polynomial is determined, the two other functions can be derived. This eventually allows us to compute the blocking probabilities at the two DCs and to estimate the efficiency of the trunk reservation policy in the framework of fog computing.

This paper is organized as follows: In Section~\ref{model}, we introduce the notation and we show convergence results for the rescaled system. We analyze in Section~\ref{rw}, the limiting random walk, in particular its kernel. The associated boundary value problems are formulated and solved in Section~\ref{bvp}. Finally, some numerical results are discussed in Section~\ref{numeric}.

\section{Model description}\label{model}
\subsection{Notation}

We consider in this paper two DCs in parallel. The first one is equipped with $ C_1 $ servers and serves customers arriving according to a Poisson process with rate $ \Lambda_1 $ and requesting exponentially distributed service times with mean $1/ \mu_1 $ (a customer if accepted occupies a single server of the DC); the number of busy servers in this first DCs is denoted by $N_1(t)$ at time $t$. Similarly, the second DCs is equipped with $ C_2 $ servers and accommodate service requests arriving according to a Poisson process with rate $ \Lambda_2 $ and service demands exponentially distributed with mean $1/ \mu_2 $; the number of requests in progress is denoted by $N_2(t)$ at time $t$. Note that the system being with finite capacity is always stable.

To reduce the blocking probability at the first service DC without exhausting the resources of the second one, we assume that when DCs 1 is full and there are at least $a$ servers available at DC 2, then requests arriving at DC 1 are accommodated by DC 2.

Figure~\ref{figModel} shows how both DCs deal with this cooperative scheme. Dashed lines represent the flows of blocked requests.

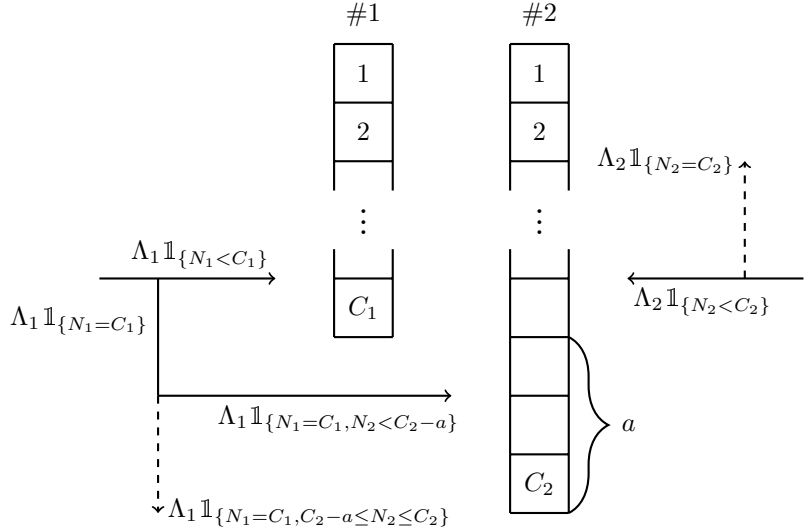
\begin{figure}[ht]
	\centering
	\begin{tikzpicture}[scale=0.78]

		\draw[thick] (1,2.5) grid +(1,2.5);
		\draw[thick] (1,0) grid +(1,1.5);		

		\draw[thick] (4,2.5) grid +(1,2.5);
		\draw[thick] (4,-3) grid +(1,4.5);

		\node at (1.5,5.5) {\#1};
		\node at (1.5,4.5) {$1$};
		\node at (1.5,3.5) {$2$};
		\node at (1.5,0.5) {$ C_1 $};
		\node at (1.5,2.1) {\huge$\vdots$};	

		\node at (4.5,5.5) {\#2};
		\node at (4.5,4.5) {$1$};
		\node at (4.5,3.5) {$2$};
		\node at (4.5,-2.5) {$ C_2 $};
		\node at (4.5,2.1) {\huge$\vdots$};
		
		\draw [decorate,thick, decoration={brace,mirror,amplitude=15pt}]
			(5.0,-3.0) -- ++ (0,3.0) node [black,midway,xshift=.8cm] {$a$};

		\draw[->, thick]
			(-3,1) -- ++	(3,0)	node[above,xshift=-1cm]	{$ \Lambda_1 \ind{N_1 < C_1 }$};
		\draw[-, thick]
			(-2,1) -- ++ (0,-2)	node[left, yshift=1cm]	{$ \Lambda_1 \ind{N_1 = C_1 }$};
		\draw[->, thick]
			(-2,-1) -- ++	(5,0)	node[below,xshift=-1.5cm]	{$ \Lambda_1 \ind{N_1 = C_1 , N_2 < C_2 - a}$};
		\draw[->, dashed, thick]
			(-2,-1) -- ++	(0,-2)	node[right]	{$ \Lambda_1 \ind{N_1 = C_1 , C_2 - a \leq N_2 \leq C_2 }$};

		\draw[->, thick]
			(9,1) -- ++	(-3,0)	node[below,xshift=1cm]	{$ \Lambda_2 \ind{N_2 < C_2 }$};
		\draw[->, dashed, thick]
			(8,1) -- ++	(0,2)	node[left]	{$ \Lambda_2 \ind{N_2 = C_2 }$};
	\end{tikzpicture}
	\caption{Policy implementation scheme} \label{figModel}
\end{figure}

Owing to the Poisson arrival and exponential service time assumptions, the process $(N(t))=((N_1(t), N_2(t)),\, t \geq 0)$ is a Markov chain, which takes values in the set $\{0, \ldots, C_1 \} \times \{0, \ldots, C_2 \}$. The transition rates of the Markov chain $(N(t))$ are given by
\begin{align*}
	q(N,N+(k,\ell)) =
	\begin{cases}
		 \Lambda_1 \ind{N_1< C_1 }	\quad & \text{if } (k,\ell)=(1,0)\\
		 \Lambda_2 \ind{N_2< C_2 } + \Lambda_1 \ind{N_1 = C_1 , N_2 < C_2 - a} \quad & \text{if } (k,\ell)=(0,1)\\
		 \mu_1 N_1 				\quad & \text{if } (k,\ell)=(-1,0)\\
		 \mu_2 N_2 				\quad & \text{if } (k,\ell)=(0,-1).
	\end{cases}
\end{align*}

In the following, we consider the process $(m(t)) = (( C_1 -N_1(t), C_2 - N_2(t)),\, t \geq 0)$, describing the number of idle servers in both DCs. In the next section, we investigate the case when the arrival rates at DCs are scaled up by a factor $\nu$.

\subsection{Rescaled system}
Let us assume that the arrival rates $ \Lambda_1 $ and $ \Lambda_2 $ are scaled up by a factor $\nu$, i.e., $ \Lambda_1 = \nu \lambda_1 $ and $ \Lambda_2 = \nu \lambda_2 $ for some factor $\nu$ and real $\lambda_1 >0$ and $ \lambda_2 >0$. We further assume that the capacities $ C_1 $ and $ C_2 $ scale with $\nu$, namely $ C_1 = \nu c_1 $ and $ C_2 = \nu c_2 $ for some positive constants $ c_1 $ and $ c_2 $. To indicate the dependence of the numbers of occupied and idle servers upon $\nu$, we write $N_i^{[\nu]}$ and $m_i^{[\nu]}$ instead of $N_i$ and $m_i$ to denote respectively the number of occupied and idle servers in DCs $i$ for $i=1,2$. 

With the above hypotheses, we are led to consider $\left(m^{[\nu]}(t)\right)$ as a random walk. For this purpose, let us introduce the random walk $(n(t)) = ((n_1(t), n_2(t)),\, t \geq 0)$ in the positive quadrant with transition rates for $n \in \N_{*}^{2}$
\begin{align*}
	r((n_1, n_2),(n_1 + k, n_2 + \ell)) =
	\begin{dcases}
		 \lambda_1 				\quad & \text{if } (k,\ell)=(-1,0)\\
		 \lambda_2 				\quad & \text{if } (k,\ell)=(0,-1)\\
		 \mu_1 c_1 				\quad & \text{if } (k,\ell)=(1,0)\\
		 \mu_2 c_2 				\quad & \text{if } (k,\ell)=(0,1)
	\end{dcases}
\end{align*}
and the reflecting conditions for $n_1 >0$ and $n_2=0$
\begin{align*}
	r((n_1,0),(n_1 + k, \ell)) =
	\begin{dcases}
		 \lambda_1 \ind{n_1>0}	\quad & \text{if } (k,\ell)=(-1,0)\\
		 \mu_1 c_1 				\quad & \text{if } (k,\ell)=(1,0)\\
		 \mu_2 c_2 				\quad & \text{if } (k,\ell)=(0,1),
\end{dcases}
\end{align*}
for $n_1=0$ and $n_2>0$
\begin{align*}
r((0, n_2),(k, n_2 + \ell)) = \begin{dcases}
		 \lambda_2 + \lambda_1 \ind{n_2>a} \qquad & \text{if } (k,\ell)=(-1,0)\\
		 \mu_1 c_1 				\quad & \text{if } (k,\ell)=(1,0)\\
		 \mu_2 c_2 				\quad & \text{if } (k,\ell)=(0,1),
	\end{dcases}
\end{align*}
and for $n = 0$
\begin{align*}
	r(0,(k, \ell)) =
	\begin{dcases}
		 \mu_1 c_1 				\quad & \text{if } (k,\ell)=(1,0)\\
		 \mu_2 c_2 				\quad & \text{if } (k,\ell)=(0,1),
	\end{dcases}
\end{align*}
where $a$ is the threshold introduced in the previous section. 

\begin{proposition}
\label{conv}
If $m^{[\nu]}(0)=(k,\ell)\in\N^2$ is fixed then, for the convergence in distribution,
$$
\lim_{\nu\to+ \infty} \left(m^{[\nu]}(t/\nu), \, t \geq 0\right)= (n(t), \, t\geq 0).
$$
\end{proposition}
\begin{proof}
If $f$ is a function on $\N^2$ with finite support then classical stochastic calculus gives the relation
\begin{align} \label{eqN}	
		& f\left(m^{[\nu]}(t/\nu)\right) = f(k,\ell) + M^{[\nu]}(t/\nu)& \\		
		&+ \int_{0}^{t} \mu_1 \frac{\nu c_1 - m_1^{[\nu]} (s/\nu)}{\nu}
		\left[f \left(m^{[\nu]}(s/\nu) + e_1 \right) - f \left(m^{[\nu]}(s/\nu) \right)\right] \diff s \notag \\
		&+ \int_{0}^{t} \mu_2 \frac{\nu c_2 - m_2^{[\nu]} (s/\nu)}{\nu}
		\left[f \left(m^{[\nu]}(s/\nu) + e_2 \right) - f \left(m^{[\nu]}(s/\nu) \right)\right] \diff s \notag \\
		&+ \int_{0}^{t} \lambda_1 \ind{m_1^{[\nu]}(s/\nu)>0}
		\left[f \left(m^{[\nu]}(s/\nu) - e_1 \right) - f \left(m^{[\nu]}(s/\nu) \right)\right] \diff s \notag\\
		&+ \int_{0}^{t} \lambda_2 \ind{m_2^{[\nu]}(s/\nu)>0}
		\left[f \left(m^{[\nu]}(s/\nu) - e_2 \right) - f \left(m^{[\nu]}(s/\nu) \right)\right] \diff s \notag\\
		&+ \int_{0}^{t} \lambda_1 \ind{{m_1^{[\nu]}(s/\nu)=0},{m_2^{[\nu]}(s/\nu)>a}} \left[f \left(m^{[\nu]}(s/\nu) - e_2 \right) - f \left(m^{[\nu]}(s/\nu) \right)\right] \diff s \notag
\end{align}
where $M^{[\nu]}(t)=\left(M_1^{[\nu]}(t),M_2^{[\nu]}(t)\right)$ is a martingale, $e_1=(1,0)$ and $e_2=(0,1)$. 

For $i=1$, $2$, since the process $\left(N_i^{[\nu]}(s/\nu)\right)$ is stochastically bounded by a Poisson process with rate $\mu_ic_i$, hence, for the convergence of processes, 
$$
\lim_{\nu\to+ \infty} \left(\frac{\nu c_i -m_i^{[\nu]}(s/\nu)}{\nu}\right)=(c_i).
$$
By using by Theorem~4.5 page~320 of Jacod and Shiryaev~\cite{Jacod}, one gets that the sequence of processes $\left(m^{[\nu]}(t/\nu),t\geq 0\right)$ is tight. If $(z(t))$ is the limit of some convergent subsequence $\left(m^{[\nu_k]}(t/\nu_k)\right)$, then, if $R=\left(r(a,b), \; (a,b) \in \N^2 \times \N^2\right)$ is the jump matrix of $(n(t))$, Relation~\eqref{eqN} gives that
$$
\left(f(z(t))-f(k,\ell)-\int_0^t R \cdot f(z(u)) \diff u \right)
$$
is a martingale. We skip some of the technical details, see Hunt and Kurtz~\cite{Hunt} for a similar context for example. This shows that $(z(t))$ is the Markov process with transition matrix $R$, hence $(z(t))$ has the same distribution as $(n(t))$. See Section~IV-20 of Rogers and Williams. Since this is the only possible limit, the convergence in distribution is proved. 
\end{proof}

By using the results of \cite{FayolleStab}, we have the following stability condition.

\begin{lemma}
A stationary regime exists for the random walk $(n(t))$ if and only if 
\begin{equation} \label{stabcond}
	 \lambda_1 > \mu_1 c_1 \mbox{ and } \mu_1 c_1 + \mu_2 c_2 < \lambda_1 + \lambda_2 .
\end{equation}
\end{lemma}

To conclude this section, let us note that the limiting random walk $(n(t))$ describes the number of customers in two $M/M/1$ queues in tandem. The arrival rate at the first queue is $ \mu_1 c_1 $ and the service rate is $ \lambda_1 $; this queue is independent of the second one and is stable under Condition~\eqref{stabcond}. The arrival rate at the second queue is $ \mu_2 c_2 $ and the service rate is $ \lambda_2 $ plus $ \lambda_1 $ when the first queue is empty and there are more than $a$ customers in the second queue. This last term introduces a coupling between the two queues. See Figure~\ref{figlimitingrw} for an illustration.

\begin{figure}[ht]
	\centering	\begin{tikzpicture}
		\draw[->]
		 (0,0) -- (6.5,0) node[below] {$n_{1}$};
		\draw[->]
		 (0,0) -- (0,6.5) node[left] {$n_{2}$};

		\node at (-.2,-.2) {$0$};
		\node at (6.2,6.2) {$\N^{2}$};
		 
		\draw[->, thick]
			(0,5.25) -- ++	(0,.75)	node[right]	{$\mu_{2} c_{2}$};
		\draw[->, thick]
			(0,5.25) -- ++	(0,-.75)	node[right]	{$\lambda_{1} + \lambda_{2}$}; 
		\draw[->, thick]
			(0,5.25) -- ++	(.75,0)	node[right]	{$\mu_{1} c_{1}$}; 

		\draw[shift={(0,4)}] (2pt,0pt) -- (-2pt,0pt) node[left] {$a$};

		\draw[->, thick]
			(0,2.75) -- ++	(0,.75)	node[right]	{$\mu_{2} c_{2}$};
		\draw[->, thick]
			(0,2.75) -- ++	(0,-.75)	node[right]	{$\lambda_{2}$}; 
		\draw[->, thick]
			(0,2.75) -- ++	(.75,0)	node[right]	{$\mu_{1} c_{1}$}; 

		\draw[->, thick]
			(3.5,0) -- ++	(.75,0)	node[above]	{$\mu_{1} c_{1}$};
		\draw[->, thick]
			(3.5,0) -- ++	(-.75,0)	node[above]	{$\lambda_{1}$}; 
		\draw[->, thick]
			(3.5,0) -- ++	(0,.75)	node[above]	{$\mu_{2} c_{2}$};

		\draw[->, thick]
			(0,0) -- ++	(.75,0)	node[above]	{$\mu_{1} c_{1}$};
		\draw[->, thick]
			(0,0) -- ++	(0,.75)	node[right]	{$\mu_{2} c_{2}$}; 

		\draw[<->, thick]
			(3,3.75)	node[left]	{$\lambda_{1}$} --	(4.5,3.75)	node[right]	{$\mu_{1} c_{1}$};
		\draw[<->, thick]
			(3.75,3)	node[below]	{$\lambda_{2}$} --	(3.75,4.5)	node[above]	{$\mu_{2} c_{2}$};
	\end{tikzpicture}
	\caption{Limiting random walk diagram}\label{figlimitingrw}
\end{figure}
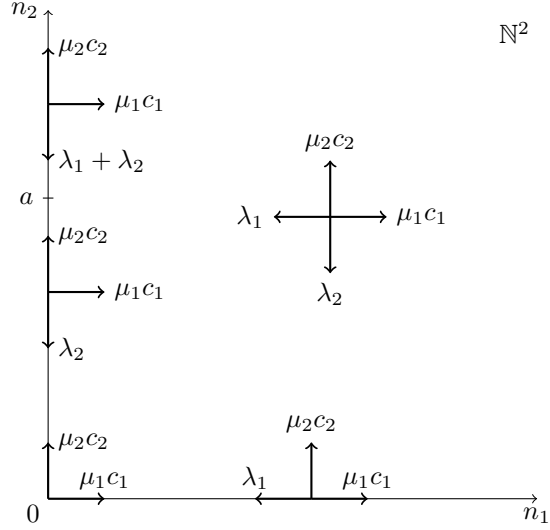

\section{Analysis of the limiting random walk}\label{rw}

We assume in this section that the random walk is ergodic. In other words, Condition~\eqref{stabcond} is satisfied.

\subsection{Functional equation}\label{funceq}

Let $p(n_1,n_2)$ denote the stationary probability of being in state $(n_1,n_2) \in \N^2$ in the stationary regime. It is easily checked that the balance equation reads for $n_1, n_2 \geq 0$
\begin{multline} \label{balance}
	\left( \lambda_1 + \lambda_2 + \mu_1 c_1 + \mu_2 c_2 - \lambda_2 \ind{n_2=0} - \lambda_1 \ind{n_1=0,0\leq n_2\leq a} \right)p(n_1,n_2) \\
	= \mu_1 c_1 p(n_1-1,n_2) + \mu_2 c_2 p(n_1,n_2-1)+ \lambda_1 p(n_1+1,n_2) \\ + \lambda_2 p(n_1,n_2+1)+ \lambda_1 p(0,n_2+1)\ind{n_1=0,n_2\geq a}.
 \end{multline}
 
Define for $x$ and $y$ in the unit disk $D = \{z \in \C: \abs{z} < 1\}$ the generating functions
\begin{align*}
	P(x,y) &= \sum_{n_1=0}^\infty \sum_{n_2=0}^\infty p(n_1,n_2) x^{n_1} y^{n_2},\\
	P_1(x) &= \sum_{n_1=0}^\infty p(n_1,0) x^{n_1}, \text{ and}\\
	P_2(y) &= \sum_{n_2=0}^\infty p(0,n_2) y^{n_2}.
\end{align*}
By definition, the function $P(x,y)$ is analytic in $D\times D$ and the functions $P_1(x)$ and $P_2(y)$ are analytic in $D$.

By multiplying Equation~\eqref{balance} by the term $x^{n_1}y^{n_2}$ and summing for $n_1$ and $n_2$ ranging from zero to infinity, we obtain the functional equation
\begin{equation}\label{eqfond}
	K(x,y)P(x,y)= \lambda_2 x(1-y)P_1(x)+ \lambda_1 (y-x)P_2(y) + \lambda_1 x(1-y)P_2^-(y),
\end{equation}
where the kernel $K(x,y)$ is defined by
\begin{equation}\label{defK}
	K(x,y) = \mu_1 c_1 x^2 y + \mu_2 c_2 x y^2 - (\lambda_1 + \lambda_2 + \mu_1 c_1 + \mu_2 c_2) x y + \lambda_1 y+ \lambda_2 x
\end{equation}
and the polynomial $P_2^-(y)$ by
$$
P_2^-(y) = \sum_{n_2=0}^{a} p(0,n_2) y^{n_2}.
$$

From the functional equation~\eqref{eqfond}, let us note that the generating function of the number of customers in the first queue is given by
$$
P(x,1) = \frac{\lambda_1}{\lambda_1 - \mu_1 c_1 x} P_2(1),
$$
from which we deduce that $P_2(1) = \P(n_1=0) = 1 - \frac{ \mu_1 c_1 }{ \lambda_1 }$. This is obvious since the first queue is an $M/M/1$ queue with input rate $ \mu_1 c_1 $ and service rate $ \lambda_1 $, independent of the second queue. See Section 6.7.1 of Robert~\cite{Robert} for a complete proof.

In addition, we have
$$
P(1,y) = \frac{ \lambda_2 P_1(1)- \lambda_1 (P_2(y)-P_2^-(y))}{ \lambda_2 - \mu_2 c_2 y}.
$$
The normalizing condition $P(1,1)=1$ yields
\begin{equation} \label{normalizing}
	 \lambda_1 \P(n_1=0,n_2\leq a)+ \lambda_2 \P(n_2=0) = \lambda_1 + \lambda_2 - \mu_1 c_1 - \mu_2 c_2 .
\end{equation}

The quantity $B_1 = \P(N_1=0,N_2\leq a)$ is the blocking probability of customers originally trying to access  the first DC and $B_2 = \P(n_2=0) $ is the blocking probability of customers trying to reach the second DC. The above relation is the global rate conservation equation of the system. The performance of the system is actually characterized by the blocking probabilities, given by the generating functions as 
\begin{equation}\label{lossgenfun}
	B = \left( P_2^{-}(1) , P_1(1) \right).
\end{equation}
In the following, we intend to compute the unknown generating functions $P_1(x)$, $P_2(y)$ and $P_2^-(y)$. 

\subsection{Zero pairs of the kernel}

The kernel $K(x,y)$ has already been studied in \cite{FayolleIas} in the framework of coupled DCs (DC). For fixed $y$, the kernel $K(x,y)$ defined by Equation~\eqref{defK} has two roots $X_0(y)$ and $X_1(y)$. By using the usual definition of the square root such that $\sqrt{a}>0$ for $a>0$, the solution which is null at the origin and denoted by $X_0(y)$, is defined and analytic in $\C\setminus ([y_1,y_2]\cup [y_3,y_4])$ where the real numbers $y_1$, $y_2$, $y_3$ and $y_4$ are such that $0<y_1<y_2<1<y_3<y_4$. The other solution $X_1(y)$ is meromorphic in $\C\setminus ([y_1,y_2]\cup [y_3,y_4])$ with a pole at 0. The function $X_0(y)$ is precisely defined by
$$
X_0(y) = \frac{-( \mu_2 c_2 y^2-( \lambda_1 + \lambda_2 + \mu_1 c_1 + \mu_2 c_2 )y+ \lambda_2 ) +{\sigma_1(y)}}{2 \mu_1 c_1 y},
$$
where $\sigma_1(y)$ is the analytic extension in $\C\setminus ([y_1,y_2]\cup [y_3,y_4])$ of the function defined in the neighborhood of 0 as $\sqrt{\Delta_1(y)}$ with
$$
\Delta_1(y) = ( \mu_2 c_2 y^2-( \lambda_1 + \lambda_2 + \mu_1 c_1 + \mu_2 c_2 )y+ \lambda_2 )^2-4 \mu_1 c_1 \lambda_1 y^2.
$$
The other solution $X_1(y) = \dfrac{ \lambda_1 }{ \mu_1 c_1 X_0(y)}$.

When $y$ crosses the segment $[y_1, y_2]$, $X_0(y)$ and $X_1(y)$ describe the circle $C_{r_1}$ with center 0 and radius $r_1=\sqrt{\frac{ \lambda_1 }{ \mu_1 c_1 }}>1$.

Similarly, for fixed $x$, the kernel $K(x,y)$ has two roots $Y_0(x)$ and $Y_1(x)$. The root $Y_0(x)$, which is null at the origin, is analytic in $\C\setminus ([x_1, x_2] \cup [x_3, x_4])$ where the real numbers $x_1$, $x_2$, $x_3$ and $x_4$ are such that with $0<x_1<x_2<1<x_3<x_4$ and is given by
$$
Y_0(x) = \frac{-( \mu_1 c_1 x^2-( \lambda_1 + \lambda_2 + \mu_1 c_1 + \mu_2 c_2 )x+ \lambda_1 ) + \sigma_2(x)} {2 \mu_2 c_2 x}
$$
where $\sigma_2(x)$ is the analytic extension in of the function defined in the neighborhood of 0 as $\sqrt{\Delta_2(x)}$ with
$$
\Delta_2(x) = (\mu_1 c_1 x^2-( \lambda_1 + \lambda_2 + \mu_1 c_1 + \mu_2 c_2 )x+ \lambda_1 )^2 - 4 \mu_2 c_2 \lambda_2 x^2.
$$
The other root $Y_1(x) = \dfrac{ \lambda_2 }{ \mu_2 c_2 Y_0(x)}$ and is meromorphic in $\C\setminus ([x_1,x_2]\cup [x_3,x_4])$ with a pole at the origin. 

When $x$ crosses the segment $[x_1, x_2]$, $Y_0(y)$ and $Y_1(y)$ describe the circle $C_{r_2}$ with center 0 and radius $r_2 = \sqrt{\frac{ \lambda_2 }{ \mu_2 c_2}}$.

\section{Boundary value problems}\label{bvp}

To solve the boundary value problems encountered in the following, let us recall that if we search for a function $P(z)$ analytic in the disk $D_r = \left\{z \in \C : \abs{z} < r \right\}$ for some $r > 0$, such that for $z \in C_r =\left\{z \in \C : \abs{z}= r \right\}$, $P(z)$ satisfies
$$
\Re(i g(z) P(z)) = \Re(ih(z))
$$
for some functions $g(z)$ and $h(z)$ analytic in a neighborhood of $C_r$ and $g(z)$ does not cancel in this neighborhood, then the function 
\begin{align*}
	\tilde{P}(z) =
	\begin{dcases}
		P(z) & z \in D_r \\ 
		\overline{P\left(1/z \right)} & z\in \C \setminus \overline{D_r}
	\end{dcases}
\end{align*}
is solution to the following Riemann-Hilbert problem: The function $\tilde{P}(z)$ is sectionally analytic with respect to the circle $C_r$ and verifies for $z\in C_r$
$$
	g(z)\tilde{P}^i(z) -\overline{g(z)}\tilde{P}^e(z) = 2i\Im(h(z)),
$$
where $\tilde{P}^i(z)$ (resp. $\tilde{P}^e(z)$) is the interior (resp. exterior) limit of the function $\tilde{P}(z)$ at the circle $C_r$, and $\overline{D_r} = D_r \cup C_r$.

From \cite{Lions}, the solution to this Riemann-Hilbert problem when it exists is given for $z \in \C \setminus C_r$ by
$$
\tilde{P}(z) = \frac{\phi(z)}{\pi}\int_{C_r} \frac{\Im(h(\xi))}{g(\xi)\phi^i(\xi)} \frac{\diff \xi}{\xi -z} + \phi(z) \mathcal{P}(z),
$$
where $\mathcal{P}(z)$ is a polynomial and $\phi^i(z)$ (resp. $\phi^e(z)$) is the interior (resp. exterior) limit at the circle $C_r$ of the solution $\phi(z)$ to the following homogeneous Riemann-Hilbert problem: The function $\phi(z)$ is sectionally analytic with respect to the circle $C_r$ and, for $z\in C_r$,
$$
\phi^i(z)= \alpha(z) \phi^e(z),
$$
where
$$
\alpha(z) = \frac{\overline{g(z)}}{g(z)}.
$$

The function $\phi(z)$ is given by
\begin{align*}
	\phi(z) =
	\begin{dcases}
		\exp\left(\frac{1}{2\pi i} \int_{C_r} {\log\left(\frac{\alpha(\xi)}{\xi^\kappa} \right)} \frac{\diff\xi}{\xi-z} \right) & z \in D_r \\ 
		\frac{1}{z^\kappa} \exp\left(\frac{1}{2\pi i} \int_{C_r} {\log\left(\frac{\alpha(\xi)}{\xi^\kappa} \right)} \frac{\diff \xi}{\xi-z}  \right) & z\in \C \setminus \overline{D_r},
	\end{dcases}
\end{align*}
where $\kappa$ is the index of the Riemann-Hilbert problem defined by 
$$
\kappa = \frac{1}{2\pi} \; \var{z \in C_r} \arg \alpha(z).
$$

The existence and the uniqueness of the solution depends on the value of the index $\kappa$. When $\kappa =0$, the solution exists and is unique with $\mathcal{P}(z) = P(0)$.

\subsection{Function $P_1(x)$}\label{FunP1}
Let us first establish a relation between the functions $P_1(x)$ and $P_2^-(y)$. For this purpose, let us define for $x \in C_{r_1}$
\begin{equation}\label{defalpha1}
	\alpha_1(x) = \frac{\bar{x}(Y_0(x)-x)}{x(Y_0(x)-\bar{x})} = \frac{ \lambda_1 (Y_0(x)-x)}{x( \mu_1 c_1 x Y_0(x)- \lambda_1 )}
\end{equation}
and let us consider the following homogeneous Riemann-Hilbert problem: The function $\phi_1(x)$ is sectionally analytic with respect to the circle $C_{r_1}$, and for some $x \in C_{r_1}$
\begin{equation}\label{RHP1hom}
	\phi_1^i(x)= \alpha_1(x) \phi_1^e(x),
\end{equation}
where $\phi_1^i(x)$ (resp. $\phi_1^e(x)$) is the interior (resp. exterior) limit at the circle $C_{r_1}$.

\begin{lemma}
The function $\phi_1(x)$ is given for some $x \in D_{r_1}$, by
\begin{equation}\label{defphi1}
	\phi_1(x) = \exp \left( \frac{x}{\pi} \int_{y_1}^{y_2} \frac{\left( \lambda_2 - \mu_2 c_2 y^2 \right) \Theta_1(y)}{y K(x,y)} \diff y \right),
\end{equation}
where
\begin{equation}\label{defTheta1}
	\Theta_1(y) = \arctan \left(\frac{\sqrt{-\Delta_1(y)}}{ \mu_2 c_2 y^2 -( \lambda_1 + \lambda_2 + \mu_1 c_1 + \mu_2 c_2 )y +2 \lambda_1 + \lambda_2 } \right).
\end{equation}
\end{lemma}

\begin{proof}
The index of the Riemann-Hilbert problem~\eqref{RHP1hom} is given by
$$
\kappa_1 = \frac{1}{2\pi} \; \var{x \in C_{r_1}} \arg \; {\alpha_1(x)}
$$
where $\alpha_1(x)$ is defined by Equation~\eqref{defalpha1}. For $x \in C_{r_1}$, we have
$$
\Re\left(\bar{x}Y_0(x)-\frac{ \lambda_1 }{ \mu_1 c_1 }\right) <0,
$$
since $Y_0(x)\in [y_1,y_2]$ and $\Re(\bar{x})\leq r_1$. Hence, $\kappa_1=0$.

The solution to the Riemann-Hilbert problem is then given for $x \in D_{r_1}$ by
$$
\phi_1(x) = \exp\left(\frac{1}{2\pi i}\int_{C_{r_1}} \frac{\log \alpha_1(z)}{z-x} \diff z \right).
$$

For $x \in C_{r_1}$, simple manipulations show that for $x = X_0(y + 0i)$ for $y \in [y_1,y_2]$
$$
\alpha_1(x) = e^{-2i\Theta_1(y)},
$$
where $\Theta_1(y)$ is defined by Equation~\eqref{defTheta1}, since
$$
x = \frac{-( \mu_2 c_2 y^2- ( \lambda_1 + \lambda_2 + \mu_1 c_1 + \mu_2 c_2 )y + \lambda_2 ) - i \sqrt{\Delta_1(y)}}{2 \mu_1 c_1 y}.
$$

It follows that for $x\in D_{r_1}$
\begin{align*}
	\frac{1}{2\pi i}\int_{C_{r_1}} \frac{\log \alpha_1(z)}{z-x} \diff z
	&=\frac{-1}{\pi}\int_{y_1}^{y_2} \Theta_1(y) \left(\frac{\frac{\diff X_0(y+0i)}{\diff y}}{X_0(y+0i)-x} + \frac{\frac{\diff X_0(y-0i)}{\diff y}}{X_0(y-0i)-x} \right) \diff y \\
	&= \frac{1}{\pi}\int_{y_1}^{y_2} \frac{x( \lambda_2 - \mu_2 c_2 y^2) \Theta_1(y) }{y K(x,y)} \diff y,
\end{align*}
where we have used the fact that
$$
\frac{\diff X_0(y+0i)}{\diff y} = \frac{X_0(y + 0i)\left(- \mu_2 c_2 y+ \frac{ \lambda_2 }{y} \right)}{-i \sqrt{-\Delta_1(y)}}
$$
and
$$
(X_0(y+0i)-x)(X_0(y - 0i) - x) = \frac{K(x,y)}{ \mu_1 c_1 y}.
$$

Equation~\eqref{defphi1} then follows.
\end{proof}

\begin{corollary}
The interior limit of the function $\phi_1(x)$ at the circle $C_{r_1}$ is given for $x=X_0(y+0i)$ with $y\in[y_1,y_2]$ by
$$
\phi^i_1(x) = \exp \left(-i\Theta_1(y) + \Phi_1(y)\right),
$$ where
\begin{equation}\label{defPhi1}
	\Phi_1(y) = \frac{1}{\pi}\oint_{y_1}^{y_2} \frac{y( \lambda_2 - \mu_2 c_2 \xi^2) \Theta_1(\xi)}{\xi (\mu_2 c_2 y \xi - \lambda_2)} \frac{\diff \xi}{\xi-y},
\end{equation}
and the symbol $\oint$ denotes the Cauchy integral \cite{Lions}.
\end{corollary}

By using Plemelj formula, we have for $x = X_0(y + 0i)$ with $y \in [y_1,y_2]$
\begin{align*}
	\phi^i_1(x)	&= \exp\left( \frac{\log \alpha_1(x)}{2} + \frac{1}{2 \pi i}\int_{C_{r_1}}\frac{\log\alpha_1(z)}{z-x} \diff z \right) \notag \\
				&= \exp\left(-i \Theta_1(y) + \frac{1}{\pi} \oint_{y_1}^{y_2} \frac{x( \lambda_2 - \mu_2 c_2 y^2) \Theta_1(\xi) }{y K(x,y)} \diff y \right). 
\end{align*}

By using the fact that for $x' = \frac{ \lambda_1 }{ \mu_1 c_1 x}$,
$$
\frac{x'}{K(x',y)} = \frac{x}{K(x,y)},
$$
we deduce that for $x = X_0(y \pm 0i)$ the Cauchy integral
$$
\frac{1}{\pi} \oint_{y_1}^{y_2} \frac{x( \mu_2 c_2 \xi^2 - \lambda_2 )\Theta_1(\xi) }{\xi K(x,\xi)} \diff \xi
$$
is real and equal to
$$
\Phi_1(y) = \frac{1}{\pi} \oint_{y_1}^{y_2} \frac{y( \lambda_2 - \mu_2 c_2 \xi^2) \Theta_1(\xi)}{\xi(\xi-y)( \mu_2 c_2 y \xi- \lambda_2 )} \diff \xi.
$$
It follows that for $x=X_0(y \pm 0i)$
$$
\phi^{i}_{1}(x) = \exp\left(\Phi_1(y) -i\Theta_1(y)\right).
$$

The results above allow us to establish a relation between functions $P_1(x)$ and $P_2^{-}(y)$.

\begin{proposition}
The functions $P_1(x)$ and $P^-_2(y)$ are such that for $x \in D_{r_1}$ 
\begin{multline} \label{P1x}
	P_1(x) = \phi_1(x) P(0,0) \\
	+ \frac{x \lambda_1 \phi_1(x)}{ \lambda_2 \pi} \int_{y_1}^{y_2} \frac{( \lambda_2 - \mu_2 c_2 y^2)P_2^-(y)}{y K(x,y)} \sin(\Theta_1(y))e^{-\Phi_1(y)} \diff y,
\end{multline}
where the functions $\phi_1(x)$ and $\Phi_1(y)$ are defined by Equations~\eqref{defphi1} and ~\eqref{defPhi1}, respectively.
\end{proposition}

\begin{proof}
For $x=X_0(y)$, we have from Equation~\eqref{eqfond}
\begin{align*}
	\frac{ \lambda_2 X_0(y)}{y-X_0(y)}P_1(x) = - \frac{X_0(y) \lambda_1 P^-_2(y)}{y-X_0(y)} -\frac{ \lambda_1 P_2(y)}{1-y},
\end{align*}

The above equation implies that for some $y \in [y_1,y_2]$, and hence $x\in C_{r_1}$, 
\begin{equation}\label{RHP1}
	\Re\left( i \lambda_2 \frac{x}{Y_0(x)-x} P_1(x) \right) = \Re\left(i \frac{x}{x-Y_0(x)} \lambda_1 P^-_2(Y_0(x)) \right),
\end{equation}
where we have used the fact that if $x=X_0(y+0i)$ with $y \in [y_1,y_2]$ then $y=Y_0(x)$.

By using the results recalled in Section~\ref{FunP1}, the solution to the Riemann-Hilbert problem~\eqref{RHP1} then reads for $x \in D_{r_1}$
\begin{align*}
	P_1(x) = \frac{\phi_1(x)}{ \lambda_2 \pi} \int_{C_{r_1}} \frac{\Im(h_1(z))}{g_1(z) \phi_1^i(z)} \frac{\diff z}{z-x} + \phi_1(x) P(0,0),
\end{align*}
where 
\begin{align*}
	h_1(x) &= \frac{x}{x-Y_0(x)} \lambda_1 P^-_2(Y_0(x)),\\
	g_1(x) &= \frac{x}{Y_0(x)-x},
\end{align*}
and the function $\phi_1(x)$ is solution to the homogeneous Riemann-Hilbert problem~\eqref{RHP1hom}.

By using the fact that $h_1(z) = -g_1(z) \lambda_1 P_2^-(Y_0(z))$, we deduce for $z=X_0(y+0i)$
$$
\frac{\Im(h_1(z))}{g_1(z) \phi_1^i(z)} = -\sin(\Theta_1(y))e^{-\Phi_1(y)} \lambda_1 P_2^-(y).
$$
Hence,
$$
\frac{1}{ \lambda_2 \pi}\int_{C_{r_1}} \frac{\Im(h_1(z))}{g_1(z) \phi_1^i(z)}\frac{\diff z}{z-x} = \frac{x \lambda_1 }{ \lambda_2 \pi} \int_{y_1}^{y_2} \frac{( \lambda_2 - \mu_2 c_2 y^2)P_2^-(y)}{y K(x,y)} \sin(\Theta_1(y))e^{-\Phi_1(y)} {\diff y}
$$
and Equation~\eqref{P1x} follows.
\end{proof}

\subsection{Function $P_2(y)$}\label{FunP2}
We now establish a relation between the function $P_2(y)$, the function $P_1(x)$ and the polynomial $P_2^-(y)$. We use the same technique as for function $P_1(x)$. 

\begin{proposition}
The function $P_2(y)$ is related to polynomial $P^-_2(y)$ as
\begin{multline}\label{P2}
	P_2(y) = \frac{y \lambda_2 }{2\pi \lambda_1 } \int_{x_1}^{x_2}\frac{P_1(x)( \lambda_1 - \mu_1 c_1 x^2) \sqrt{-\Delta_2(x)} }{x ( \lambda_1 + \lambda_2 - ( \mu_1 c_1 + \mu_2 c_2 )x)K(x,y)} {\diff x} \\
			+ \frac{y}{2 \pi i} \int_{C_{r_2}} \frac{(z-1)( \lambda_2 - \mu_2 c_2 z^2)X_0(z)^2P_2^-(z)}{z^2(z-X_0(z))K(X_0(z),y)} {\diff z} + P(0,0).
\end{multline}
\end{proposition}

\begin{proof}
From Equation~\eqref{eqfond}, we have for $y=Y_0(x)$
$$
 \lambda_2 \frac{1-Y_0(x)}{(Y_0(x)-x)} x P_1(x) + \lambda_1 P_2(Y_0(x)) + \lambda_2 \frac{1-Y_0(x)}{(Y_0(x)-x)} x P^-_2(Y_0(x)) = 0.
$$
When $x\in [x_1,x_2]$, $Y_0(x)\in C_{r_2}$ and it then follows that the function $P_2(y)$ satisfies for $y \in C_{r_2}$
$$
\Re\left(i \lambda_1 P_2(y) \right) = \Re\left(i h_2(y) \right),
$$
where
$$
h_2(y) = \lambda_2 \frac{(y-1)X_0(y)}{y-X_0(x)}P_1(X_0(y))+ \lambda_1 \frac{(y-1)X_0(y)}{y-X_0(y)}P_2^-(y).
$$
By using the results recalled in Section~\ref{FunP1}, the function $P_2(y)$ is given by
$$
P_2(y) = \frac{1}{ \lambda_1 \pi}\int_{C_{r_2}} {\Im(h_2(z))}\frac{\diff z}{z-y}+P(0,0).
$$
(Note that we have in the present case to deal with a Dirichlet problem.)

Simple computations show that for $y = Y_0(x+0i)$
$$
\Im\left( \lambda_2 \frac{(y-1)X_0(y)}{y-X_0(x)}P_1(X_0(y)) \right)=-\frac{ \lambda_2 P_1(x)}{2( \lambda_1 + \lambda_2 -( \mu_1 c_1 + \mu_2 c_2 )x)}\sqrt{-\Delta_2(x)}.
$$
In addition, by using the fact that the polynomial $P_2^-(y)$ is with real coefficients, we have
\begin{multline*}
	\Im\left( \lambda_1 \frac{y-1}{y-X_0(y)}X_0(y)P_2^-(y) \right)=\frac{ \lambda_1 }{2 \pi i} \left(\frac{y-1}{y-X_0(y)}X_0(y)P_2^-(y)
	\right. \\ \left.
	-\frac{\bar{y}-1}{\bar{y}-X_0(y)}X_0(y)P_2^-(\bar{y}) \right).
\end{multline*}

Since
$$
\frac{\diff Y_0(x+0i)}{\diff x}= \frac{Y_0(x+0i) \left(\frac{ \lambda_1 }{x} - \mu_1 c_1 x\right)}{-i \sqrt{-\Delta_2(x)}},
$$
we have
\begin{multline*}
	\int_{C_{r_2}} \Im\left( \lambda_2 \frac{(y-1)X_0(y)}{y-X_0(x)}P_1(X_0(y)) \right) \frac{\diff z}{z-y} = \\
	\int_{x_1}^{x_2}\frac{{ \lambda_2 y} P_1(x)( \lambda_1 - \mu_1 c_1 x^2) \sqrt{-\Delta_2(x)} }{2x( \lambda_1 + \lambda_2 -( \mu_1 c_1 + \mu_2 c_2 )x)K(x,y)} {\diff x}.
\end{multline*}
Moreover,
\begin{align*}
	&\frac{1}{\pi} \int_{C_{r_2}} \Im\left( \lambda_1 \frac{z-1}{z-X_0(z)}X_0(z)P_2^-(z) \right) \frac{\diff z}{z-y}\\
	&= \frac{y \lambda_1 }{2 \pi i} \int_{C_{r_2}} \frac{(z-1)( \lambda_2 - \mu_2 c_2 z^2)X_0(z)P_2^-(z)}{z(z-X_0(z))(z-y)( \lambda_2 - \mu_2 c_2 yz)}{\diff z} \\
	&= \frac{y \lambda_1 }{2 \pi i} \int_{C_{r_2}} \frac{(z-1)( \lambda_2 - \mu_2 c_2 z^2)X_0(z)^2P_2^-(z)}{z^2(z-X_0(z)) K(X_0(z),y)} {\diff z}.
\end{align*}
By assembling the two above relations, Equation~\eqref{P2} follows.
\end{proof}

\subsection{Determination of the polynomial $P_2^{-}(y)$}
We use the two previous results to establish a linear system satisfied by the coefficients of the polynomial $P_2^{-}(y)$. Let us first introduce some notation

Set
$$ x(\theta) = \frac{ \lambda_1 + \lambda_2 + \mu_1 c_1 + \mu_2 c_2 -2\sqrt{ \mu_2 c_2 \lambda_2 }\cos\theta - \sqrt{\delta_1(\theta)}}{2 \mu_1 c_1 }
$$
with
$$ \delta_1(\theta) = ( \lambda_1 + \lambda_2 + \mu_2 c_2 + \mu_1 c_1 -2\sqrt{ \mu_2 c_2 \lambda_2 }\cos\theta)^2-4 \mu_1 c_1 \lambda_1 ,
$$
so that
$$
\cos\theta = - \frac{\mu_1 c_1 x(\theta)^2 - ( \lambda_1 + \lambda_2 + \mu_1 c_1 + \mu_2 c_2 ) x(\theta)+ \lambda_1 }{2\sqrt{ \mu_2 c_2 \lambda_2 } x(\theta)}.
$$

Moreover, define the coefficients for $n,k = 0, \ldots, a$
\begin{align}
	\alpha_{n,k}^{(1)}	&= \frac{2 \mu_2 c_2 }{\pi^2}\int_0^\pi \frac{ x(\theta)\phi_1(x(\theta))}{ \lambda_1 + \lambda_2 -( \mu_1 c_1 + \mu_2 c_2 )x(\theta)} J_k(\theta)\sin((n+1)\theta)\sin\theta {\diff \theta}, \label{alpha1}\\
	\alpha_{n,k}^{(2)}	&= \frac{2 r_2^{k-1} \mu_2 c_2 }{ \pi} \int_0^\pi x(\theta)j_k(\theta) \sin((n+1)\theta) {\diff \theta}, \label{alpha2} \\
	\beta_{n}			&= \frac{2 \mu_2 c_2 \lambda_2 }{ \lambda_1 \pi}\int_{0}^\pi \frac{\phi_1(x(\theta))x(\theta)}{ \lambda_1 + \lambda_2 -( \mu_1 c_1 + \mu_2 c_2 )x(\theta)}\sin((n+1)\theta) \sin\theta {\diff \theta},\label{defbetan}
\end{align}
where
\begin{align*}
	J_k(\theta) &= \int_{y_1}^{y_2} \frac{( \lambda_2 - \mu_2 c_2 y^2)y^{k-1}}{ \mu_2 c_2 y^2 -2y\sqrt{ \mu_2 c_2 \lambda_2 }\cos\theta + \lambda_2 } \sin(\Theta_1(y))e^{-\Phi_1(y)} {\diff y} ,\\
	j_k(\theta) &= \frac{(r_2^2+x(\theta))\sin(k\theta)-x(\theta)r_2\sin((k+1)\theta)-r_2\sin((k-1)\theta)}{(1-x(\theta))( \lambda_1 + \lambda_2 -( \mu_1 c_1 + \mu_2 c_2 )x(\theta))}.
\end{align*}

\begin{theorem}\label{Theorem1}
The probabilities $p(0,n)$ for $n = 1, \ldots, a$ can be expressed as a function of $p(0,0)$ by solving the linear system: For $n=0,\ldots, a-1$
\begin{equation}\label{linsys}
	r_2^n p(0,n+1) =\beta_{n} p(0,0)+ \sum_{k=0}^a \alpha_{n,k} p(0,k)
\end{equation}
where $\beta_n$ is given by Equation~\eqref{defbetan} and $\alpha_{n,k}= \alpha_{n,k}^{(1)}+ \alpha_{n,k}^{(2)}$ with $\alpha_{n,k}^{(1)}$ and $\alpha_{n,k}^{(2)}$ being defined by Equation~\eqref{alpha1} and \eqref{alpha2}, respectively. The probability $p(0,0)$ is determined by plugging the solution the linear system into equation~\eqref{P2}, and such that it satisfies the normalization condition~\eqref{normalizing}.
\end{theorem}

\begin{proof}
By using the fact that
$$
\frac{1}{K(x,y)} = \frac{1}{ \lambda_2 x}\sum_{n=0}^\infty U_n\left(\cos\theta_2(x)\right)\left(\frac{ y}{r_2} \right)^n,
$$
where $U_n(x)$ is the $n$th Chebyshev polynomials of the second kind \cite{Abramowitz}, $\theta_2(x) \in [0,\pi]$ is such that for $x\in [x_1,x_2]$
\begin{align*}
	\cos\theta_2(x) &= -\frac{ \mu_1 c_1 x^2-( \lambda_1 + \lambda_2 + \mu_1 c_1 + \mu_2 c_2 )x+ \lambda_1 }{2x\sqrt{ \mu_2 c_2 \lambda_2 }},\\
	\sin\theta_2(x) &= \frac{\sqrt{-\Delta_2(x)}}{2x \sqrt{ \mu_2 c_2 \lambda_2 }}
\end{align*}
and then $Y_0(x+0i) =\sqrt{\frac{ \lambda_2 }{ \mu_2 c_2 }}e^{-i\theta_2(x)}$,
we deduce from Equation~\eqref{P2} that
\begin{align*}
	r_2^n p(0,n+1) &= \\
	& \frac{1}{2\pi \lambda_1 } \int_{x_1}^{x_2}\frac{P_1(x) ( \lambda_1 - \mu_1 c_1 x^2) \sqrt{-\Delta_2(x)} }{x^2( \lambda_1 + \lambda_2 -( \mu_1 c_1 + \mu_2 c_2 )x)} U_n\left(\cos\theta_2(x)\right) {\diff x} \\
	&+ \frac{1}{2 \pi i \lambda_2 } \int_{C_{r_2}} \frac{(z-1)( \lambda_2 - \mu_2 c_2 z^2)X_0(z)P_2^-(z)}{z^2(z-X_0(z))} U_n\left(\cos\theta_2(X_0(z))\right) {\diff z}.
\end{align*}

We have
\begin{multline*}
	\frac{1}{2 \pi i \lambda_2 } \int_{C_{r_2}} \frac{(z-1)( \lambda_2 - \mu_2 c_2 z^2)X_0(z)P_2^-(z)}{z^2(z-X_0(z))} U_n\left(\cos\theta_2(X_0(z))\right) {\diff z} = \\
	\sum_{k=0}^a p(0,k) \frac{1}{2 \pi i \lambda_2 } \int_{C_{r_2}} \frac{(z-1)( \lambda_2 - \mu_2 c_2 z^2)X_0(z) z^k}{z^2(z-X_0(z))} U_n\left(\cos\theta_2(X_0(z))\right) {\diff z}.
\end{multline*}
For the integrals appearing in the above equation, we note that
\begin{multline*}
	\frac{1}{2 \pi i \lambda_2 } \int_{C_{r_2}} \frac{(z-1)( \lambda_2 - \mu_2 c_2 z^2)X_0(z) z^k}{z^2(z-X_0(z))} U_n\left(\cos\theta_2(X_0(z))\right) {\diff z} = \\
	-\frac{1}{\pi r_2} \int_{C_{r_2}}\frac{(z(\theta)-1)X_0(z(\theta))}{z(\theta)-X_0(z(\theta))} z(\theta)^{k-1}\sin((n+1)\theta) {\diff \theta},
\end{multline*}
where $z(\theta)=r_2e^{i\theta}$ and where we have used the fact that
$$
U_n(\cos\theta)=\frac{\sin((n+1)\theta)}{\sin \theta}.
$$
Simple computations show that the integral in the right hand side of the above equation is equal to $\alpha^{(2)}_{n,k}$ defined by Equation~\eqref{alpha2}.

By using the fact that 
$$
2\sqrt{ \mu_2 c_2 \lambda_2 }\sin\theta_2(x) \frac{\diff \theta_2(x)}{\diff x}= \frac{ \mu_1 c_1 x^2- \lambda_1 }{x^2}
$$
we easily obtain that 
\begin{multline*}
	\frac{1}{2\pi \lambda_1 } \int_{x_1}^{x_2}\frac{P_1(x) ( \lambda_1 - \mu_1 c_1 x^2) \sqrt{-\Delta_2(x)}}{x^2( \lambda_1 + \lambda_2 - ( \mu_1 c_1 + \mu_2 c_2 ) x)} U_n\left(\cos\theta_2(x)\right) {\diff x} = \\
	\frac{2 \mu_2 c_2 \lambda_2 }{\pi \lambda_1 } \int_{0}^{\pi}\frac{P_1(x(\theta)) x(\theta) }{\lambda_1 + \lambda_2 - ( \mu_1 c_1 + \mu_2 c_2 )x(\theta)} \sin((n+1)\theta)\sin\theta {\diff \theta}
\end{multline*}

Simple manipulations show that the above integral can be expressed as
$$
\beta_n p(0,0) + \sum_{k=0}^a \alpha_{n,k}^{(1)} p(0,k)
$$
with $\alpha_{n,k}^{(1)}$ and $\beta_n$ defined by Equations~\eqref{alpha1} and~\eqref{defbetan}, respectively. We finally obtain the linear system~\eqref{linsys}.
\end{proof}

In the next section, we show how the above theorem can be used to carry out numerical experiments.

\section{Numerical experiments}\label{numeric}

We report in this section some numerical results in order to illustrate the computations performed in the previous sections and the effects of the implementation of a (unilateral) \emph{trunk reservation} in a fog system composed by two large data centers. We focus on the blocking probabilities at each DC.

In all cases described below, it is considered that DC 1 is operating under overload conditions, $ \lambda_1 > \mu_1 c_1 $ and Condition~\eqref{stabcond} is always assumed. Thus, recalling the notation of $B_i$ as the loss rate of customers first assigned to DC $i$, for $i = 1,2$, we have
\begin{enumerate}[i.]
\item If $\mu_1 c_1 + \mu_2 c_2 < \lambda_1 + \lambda_2$ and $a \geq 1$, the loss probabilities are calculated by using Theorem~\ref{Theorem1}.
It is worth noticing that if $a$ assumes very large values, the system performs as uncoupled DCs. This limiting result $B^{\infty}$ is given by
\begin{align} \label{Bind}
	B_1^{\infty} = 1 - \frac{\mu_1 c_1}{\lambda_1} \quad\text{and} \quad
	B_2^{\infty}=\begin{dcases}
		0							&\text{if } \lambda_2 < \mu_2 c_2 \\
		1 - \mu_2 c_2 / \lambda_2 	&\text{if } \lambda_2 > \mu_2 c_2.
	\end{dcases}
\end{align}
\item If $ \lambda_1 + \lambda_2 > \mu_1 c_1 + \mu_2 c_2 $ and $a = 0$, there is no protection for the requests that originally arrive at DC 2; every request blocked at DC 1 is systematically forwarded to DC 2, regardless of the amount of idle servers in this DC. This scenario is a particular case of the offloading scheme studied in \cite{FGRT}, where 2 parallel DCs forward demands they cannot hold to each other with a given probability. Namely, in the notation of that paper, we have $p_1 = 1$ and $p_2 = 0$, where $p_i$ is probability of DC $i$ forward a request to DC $3-i$, for $i = 1,2$.

By using \cite[Theorem 3]{FGRT}, we can calculate the quantity $B^0$, using
\begin{align} \label{Bsha}
	B_1^0 = P(0,0) \qquad \text{and} \qquad 
	B_2^0 = \frac{\lambda_2 + \lambda_1 (1-P(0,0)) - \mu_1 c_1 - \mu_2 c_2}{\lambda_2}.
\end{align}
where
\begin{align*}
	P(0,0) = 
	\begin{dcases}
		\frac{ \lambda_1 - \mu_1 c_1 }{ \lambda_1 \phi_2(1)} & \textit{if } \lambda_2 > \mu_2 c_2  \\
		\frac{ \lambda_1 + \lambda_2 - \mu_1 c_1 - \mu_2 c_2 }{ \lambda_1 \phi_2(1)} & \textit{if } \lambda_2 < \mu_2 c_2 ,
	\end{dcases}
\end{align*}
with,
\begin{align*}
	\phi_2(y) = \exp{\left(\dfrac{y}{\pi} \int_{x_1}^{x_2} \dfrac{\lambda_1 - \mu_1 c_1 x^2}{x K(x,y)} \Theta_2(x) {\diff x} \right)}
\end{align*}
and
\begin{align*}
	\Theta_2(x) = \mathrm{ArcTan} \left(\dfrac{\sqrt{-\Delta_2(x)}} {(x+1)( \lambda_1 - x \mu_1 c_1 ) + x ( \lambda_2 - \mu_2 c_2 )} \right).
\end{align*}
\end{enumerate}

To realize the effects of the choice of the parameter $a$ on the loss probabilities of both DCs, we compare the blocking rates in this system with the intuitive boundaries, $B^0$ and $B^{\infty}$, defined previously in Equations~\eqref{Bind}~and~\eqref{Bsha} respectively.

In Figure~\ref{Figure3}, we show the loss probabilities of each DC, for the case where it is assumed that both DC have the same capacity, taken as unity and DC 2 is twice as fast as DC 1, whose mean service time is taken as the unity. Additionally, the arrivals occur at rate $3$ and $5$ at DCs 1 and 2 respectively.
\begin{figure}[htbp]
	\centering
	\includegraphics[scale=0.7]{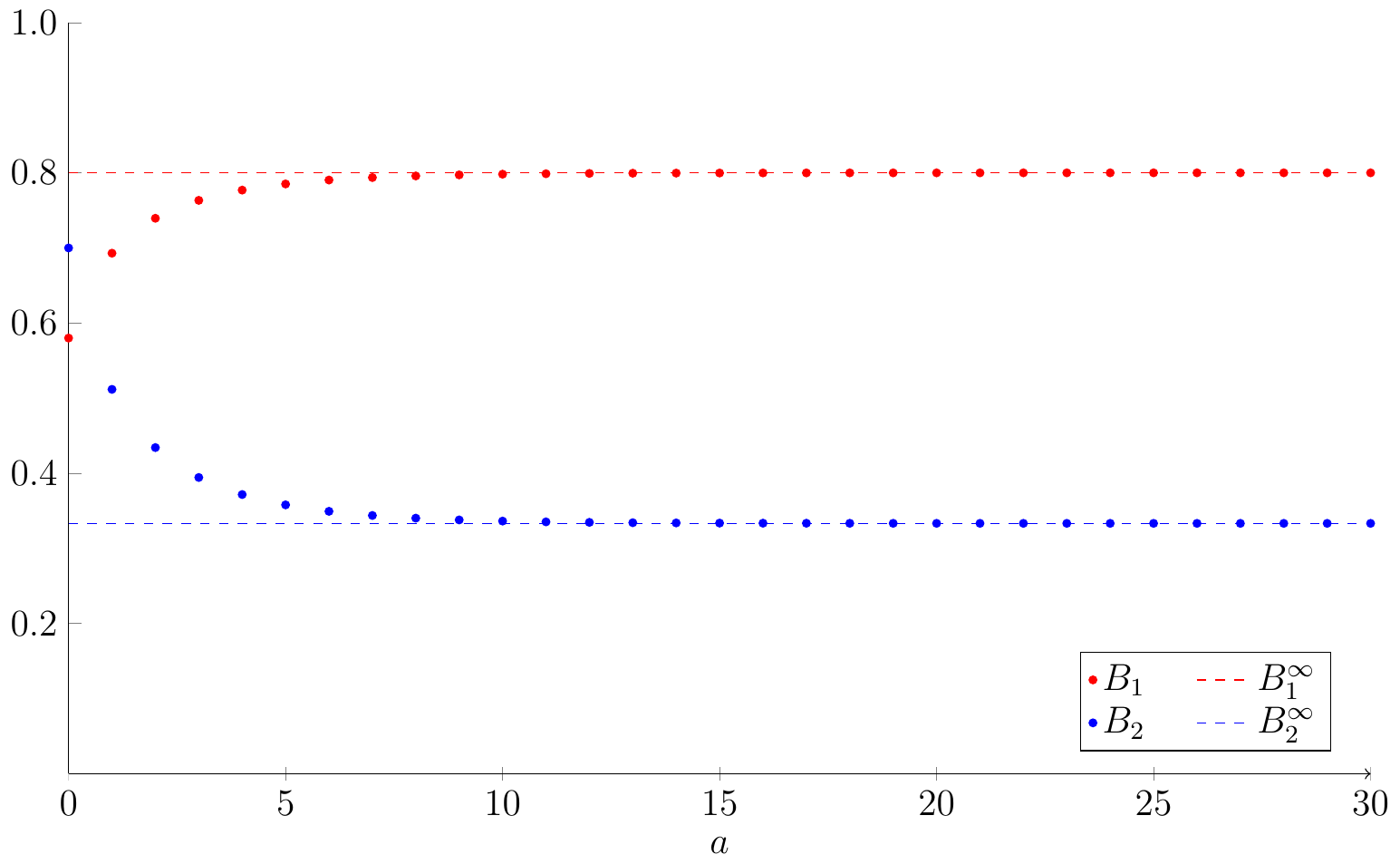}	
	\caption[caption]{Blocking rates for \\ $\lambda = (3.,5.)$, $\mu = (1.,2.)$ and $c= (1.,1.)$} \label{Figure3}
\end{figure}

We see from Figure~\ref{Figure3} that small values of the threshold $a$ can be used to decrease the blocking rate of requests originally assigned to the second DC without increasing by too much the rejection of forwarded requests. We notice that for both DCs, the blocking rates approach very quickly to their limits $B^\infty$, since DC 2 is highly loaded, so it is expected that any reservation near the full occupancy will quickly hinder their cooperative scheme, isolation is obtained as $a$ grows.

In Figures~\ref{Figure4}~and~\ref{Figure5}, we compare the implementation of this policy in an offloading perspective, where we think of DC 2 as a large backup  facility for the operation of DC 1. It is assumed now that the DCs have the same service rate, taken as the unity, but DC 2 has 10 times the size of DC 1. We contrast the cases, where in Figure~\ref{Figure4} DC 2 is operating originally underloaded, $\mu_2 c_2 > \lambda_2$, and in Figure~\ref{Figure5} both DCs are overloaded.
\begin{figure}[htbp]
	\centering
	\includegraphics[scale=0.7]{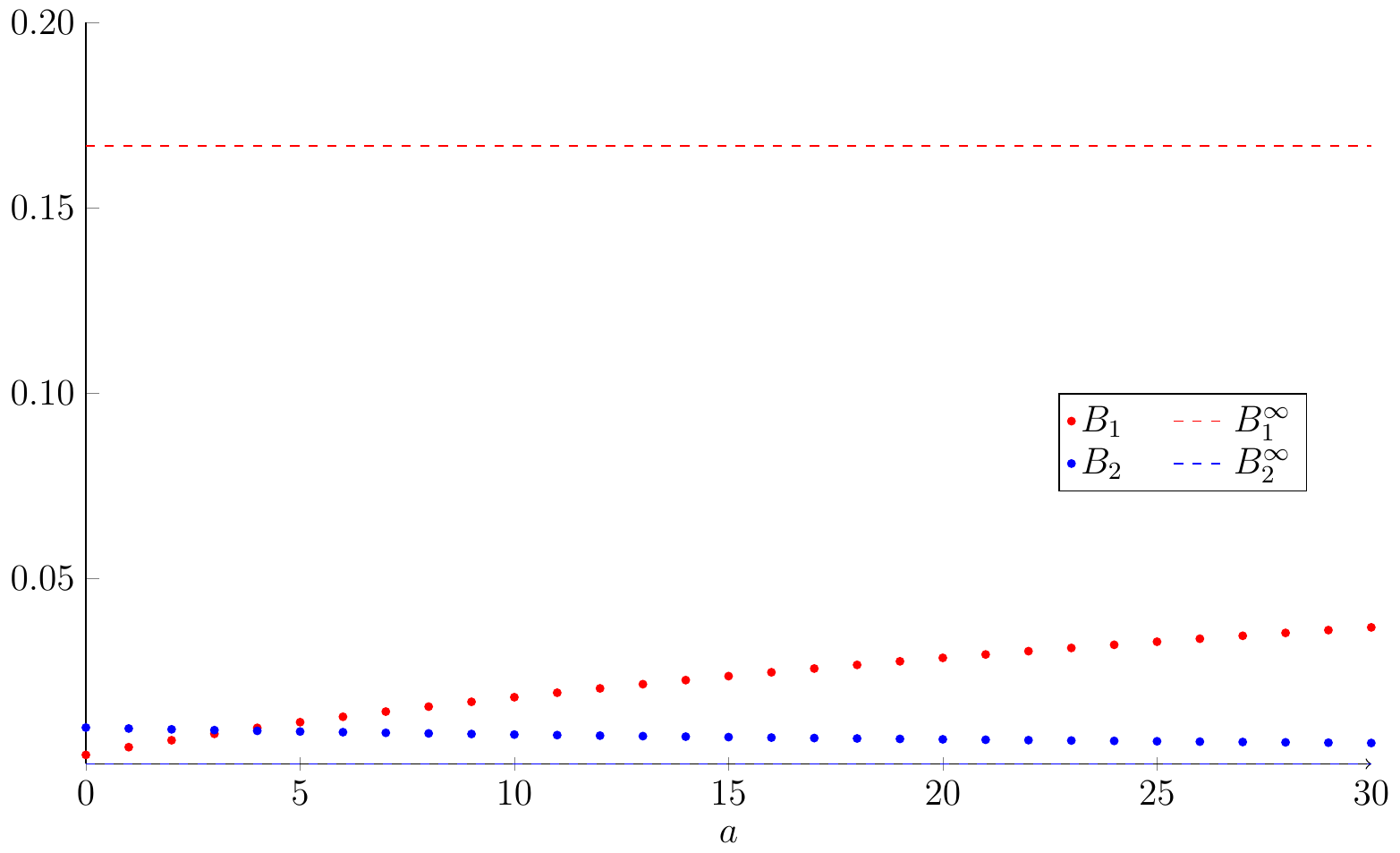}	
	\caption[caption]{Blocking rates for \\ $\lambda = (1.2,9.9)$, $\mu = (1.,1.)$ and $c= (1.,10.)$}\label{Figure4}
\end{figure}

\begin{figure}[htbp]
	\centering
	\includegraphics[scale=0.7]{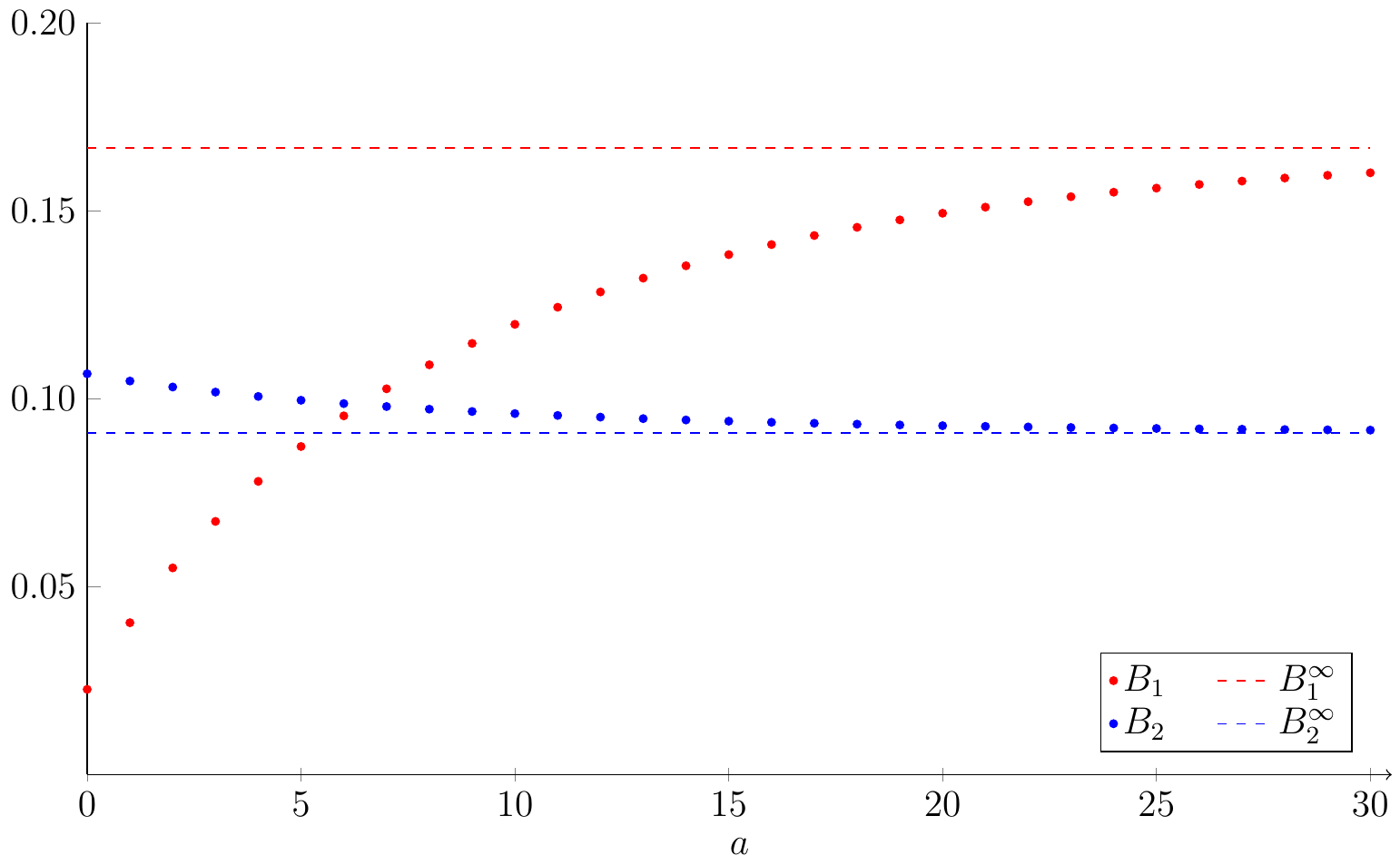}	
	\caption[caption]{Blocking rates for \\ $\lambda = (1.2,11.)$, $\mu = (1.,1.)$ and $c= (1.,10.)$}\label{Figure5}
\end{figure}

It is visible that the trunk reservation  policy achieves some loss reduction on the first DC (compared with the total isolation case) without jamming the service for those requests which are originally assigned to the second DC. In both cases, this policy enables sustainable cooperation, where DC 2 could share its resources without been over penalized, particularly because of its size.

However, when capacities $c_1$ and $c_2$ are similar, the scheme introduces little improvement, especially when the second DC is highly loaded. In the context of fog computing, this policy enables a decentralized cooperative scheme between the edge and core of the cloud. The results obtained in this paper gives some hints to optimize the system, using a distributed routine by the implementation of a saving threshold that can enable better service level and improve network usage.

\section*{Acknowledgement}

The authors thank Philippe Robert for the proof of Proposition~\ref{conv}.

\bibliographystyle{plain}
\bibliography{walk}

\end{document}